\documentclass[letterpaper, 10 pt, conference]{ieeeconf}
\IEEEoverridecommandlockouts
\renewcommand{\baselinestretch}{0.9}

\usepackage{amsmath,amssymb,amsfonts}
\usepackage{algorithm}
\usepackage{algorithmic}
\usepackage{graphicx}
\usepackage{textcomp}
\usepackage{xcolor}
\usepackage{epstopdf}
\usepackage{balance}
\usepackage{ulem}
\usepackage{fancyhdr}
\usepackage[noadjust]{cite}
\usepackage{verbatim}
\allowdisplaybreaks[4]
\def\BibTeX{{\rm B\kern-.05em{\sc i\kern-.025em b}\kern-.08em
    T\kern-.1667em\lower.7ex\hbox{E}\kern-.125emX}}
    
\newtheorem{assumption}{Assumption}

\newtheorem{remark}{Remark}
\newtheorem{theorem}{Theorem}
\newtheorem{lemma}{Lemma}

\newtheorem{definition}{Definition}

\newtheorem{problem}{Problem}

\fancypagestyle{myheader}{
    \fancyhf{} 
    \fancyhead[L]{ } 
    \fancyhead[C]{\footnotesize This work has been submitted to IEEE CDC 2025. Copyright may be transferred without notice, after which this version might no longer be accessible.} 
    \fancyhead[R]{ } 
    \setlength{\headheight}{0pt} 
    \setlength{\headsep}{25pt}    
}
\title{Distributed Finite-Horizon Optimal Control for Consensus with\\ Differential Privacy Guarantees}

\clearpage

\author{Yuwen Ma\textsuperscript{1}, Yongqiang Wang\textsuperscript{2}, Sarah K. Spurgeon\textsuperscript{1}, Boli Chen\textsuperscript{1}
\thanks{This work was supported by Engineering and Physical Sciences Research Council (EPSRC) of UK Research and Innovation (UKRI) under Grant EP/W524335/1.}
\thanks{\textsuperscript{1}Yuwen Ma, Sarah K. Spurgeon and Boli Chen are with the Department of Electronic and Electrical Engineering, University College London, London WC1E 7JE, UK, {\tt\small \{yuwen.ma.24\}, \{s.spurgeon\}, \{boli.chen\}@ucl.ac.uk}
        }
\thanks{\textsuperscript{2}Yongqiang Wang is with the Department of Electrical and Computer Engineering, Clemson University, Clemson, SC 29634, USA, {\tt\small yongqiw@clemson.edu}}
}

\begin{document}
\pagestyle{myheader}
\begin{titlepage}
    \centering
    \vspace*{\fill}
    {\Huge \bfseries Copyright Statement \\[0.5cm]}

    {\large This work has been submitted to the IEEE Conference on Decision and Control 2025. Copyright may be transferred without notice, after which this version might no longer be accessible. \\[2cm]}
    \vspace*{\fill}
\end{titlepage}
\maketitle
{ 
\begin{abstract}
This paper addresses the problem of privacy-preserving consensus control for multi-agent systems (MAS) using differential privacy. We propose a novel distributed finite-horizon linear quadratic regulator (LQR) framework, in which agents share individual state information while preserving the confidentiality of their local pairwise weight matrices, which are considered sensitive data in MAS. Protecting these matrices effectively safeguards each agent's private cost function and control preferences. Our solution injects consensus error-dependent Laplace noise into the communicated state information and employs a carefully designed time-dependent scaling factor in the local cost functions. {This approach guarantees bounded consensus and achieves rigorous $\epsilon$-differential privacy for the weight matrices without relying on specific noise distribution assumptions.} Additionally, we analytically characterize the trade-off between consensus accuracy and privacy level, offering clear guidelines on how to enhance consensus performance through appropriate scaling of the LQR weight matrices and the privacy budget.
\end{abstract}
}


\section{Introduction}
{ The emergence of large-scale network systems, such as sensor networks, intelligent transportation systems, and power grids, has driven the need for distributed control and optimization algorithms \cite{zhao2023distributed,shi2012reaching}. 
With the increasing scale of data in the network, the frequent exchange of data amplifies the risk of adversaries intercepting transmitted data and extracting sensitive information, such as references, control preferences, and individual system dynamics. 
Such vulnerabilities significantly threaten system reliability, underlining the critical importance of securing transmitted data \cite{han2016differentially}. A widely adopted method for protecting sensitive information involves injecting noise into the exchanged data. However, this approach inherently reduces control performance. Differential privacy, introduced by Dwork, provides a rigorous metric to characterize and manage this privacy-performance trade-off effectively \cite{dwork2006calibrating,dwork2014algorithmic}. Compared to encryption methods, which typically offer higher accuracy \cite{gentry2009fully}, differential privacy reduces communication overhead and provides a quantifiable measure of privacy protection.}
 
 In light of these considerations, recent research has increasingly focused on protecting transmitted data using the differential privacy mechanism \cite{huang2015differentially,huang2024differential,han2016differentially,wang2024robust,xuan2023gradient,cao2020differentially,zellner2017distributed,yazdani2022differentially,cortes2016differential}. The first exploration of cost function protection appeared in \cite{huang2015differentially}, where a fundamental trade-off between privacy level and accuracy level is established when a graph is undirected. Later, the work in \cite{huang2024differential} further relaxed the assumption of an undirected graph by adding noise to decision variables and to the estimate of the aggregated gradient to preserve cost function privacy. With guaranteed privacy for both cost functions and constraint sets, a consensus problem is solved by a fully distributed optimization algorithm in \cite{wang2024robust}. Moreover, the proposed algorithm achieves rigorous $\epsilon$-differential privacy and accurate convergence for the first time. Further advancements in differentially private distributed optimization and control algorithms are presented in \cite{cao2020differentially,zellner2017distributed,yazdani2022differentially}.
 
{ 
In this paper, we investigate the design of differentially private distributed finite-horizon linear quadratic regulators (LQR) to achieve consensus in multi-agent systems (MAS). Our primary goal is to protect the quadratic cost function from adversarial manipulation, preventing agents from minimizing their own costs at the expense of others or causing system miscoordination. We focus specifically on safeguarding the pairwise weighting matrices, which sufficiently reflect the operational priorities expressed through the cost function. Protecting these matrices directly often yields superior outcomes compared to methods that protect the entire cost function, as seen in existing studies \cite{huang2015differentially,huang2024differential,wang2024robust}. Such comprehensive protection approaches typically result in higher global sensitivity and, consequently, a more conservative privacy-performance trade-off. 
Unlike many existing methods that inject noise under the assumption that the noise variance or its expectation remains bounded to ensure bounded control performance \cite{huang2015differentially,huang2024differential,wang2024robust}, our proposed algorithm employs a carefully tailored Laplace distribution conditioned explicitly on the consensus error. This design removes the need for bounded noise assumptions without compromising closed-loop performance, primarily due to the introduction of a novel scaling factor embedded within the cost function. Using Lyapunov analysis, we demonstrate that our approach achieves bounded consensus error while guaranteeing $\epsilon$-differential privacy, even as the number of iterations increases indefinitely. The analytical results characterize the inherent trade-off between consensus accuracy and privacy level, highlighting a distinct accuracy improvement over existing approaches \cite{huang2015differentially}. Furthermore, our results offer clear guidelines on enhancing consensus performance through the strategic scaling of LQR weight matrices.
}

 
\textbf{Notation:} ${\mathbb{R}^{m\times n}}$ and ${\mathbb{R}^{n}}$ are the sets of all ${m\times n}$-dimensional real matrices and $n$-dimensional real vectors, respectively. $\mathbb{N}$ denotes the set of natural numbers, while $\mathbb{N}_+$ represents the set of positive integers. ${\mathbb{S}^{n\times n}}$ and ${\mathbb{S}_+^{n\times n}}$ denote the set of all $n\times n$-dimensional positive semi-definite matrices and positive-definite matrices, respectively. ${\mathbf{I}}$ is an identity matrix with appropriate dimensions, $\mathbf{0}$ is a zero matrix, and ${\mathbf{1}}$ is a column vector of ones. $\mathrm{diag}\left\{ {{A}_{1}},\ldots, {{A}_{n}}\right\}$ is a diagonal matrix with ${{{A}_{1}},\ldots,{{A}_{n}}}$ as the diagonal elements. {$\mathrm{range}(\cdot)$} indicates the range of an operator. ${\left\| \cdot \right\|_\infty}$ stands for the $\ell_\infty$-norm. Unless otherwise specified, $\|\cdot\|$ represents the \( \ell_2 \)-norm. Given a square matrix $A\in \mathbb{R}^{n\times n}$, $\lambda_{0}(A)$ denotes the smallest positive eigenvalue, and { $\mathrm{tr}(A)$ represents its trace.} The mathematical expectation is denoted by $\mathbb{E}[\cdot]$. $\mathbb{P}(B)$ is the probability of event $B$.

\section{Problem Formulation}

Consider a multi-agent system of $N$ agents with a discrete-time single integrator model:
\begin{equation}\label{e1}
    x_i(t+1)=x_i(t)+u_i(t),\ i=1,2,\ldots,N, \ t\in \mathbb{N},
\end{equation}
where $x_i(t)\in \mathbb{R}^n$ and $u_i(t)\in \mathbb{R}^n$ represent the state and control input, respectively. The initial state $x_i(0)$ is given for all $i=1,2,\ldots,N$.

There exists a communication graph $\mathcal{G}(\mathcal{V} ,\mathcal{E},\mathcal{A} )$ among the above $N$-agent system. The graph $\mathcal{G}$ is described as a directed graph consisting of a node set $\mathcal{V}=\{1,2,\ldots,N\}$, an edge set $\mathcal{E} \subseteq \mathcal{V} \times \mathcal{V}$ and an adjacency matrix $\mathcal{A}=\left[\begin{matrix}a_{ij}\end{matrix}\right]\in \mathbb{R}^{N\times N}$. For $\left( j,i \right)\in \mathcal{E}$, we mean that there is a link from node $j$ to node $i$, or equivalently node $j$ is an in-neighbor of node $i$. Let ${{\mathcal{N}}_{i}= \{ {j\in \mathcal{V} |(j,i)\in \mathcal{E} }\}}$ denote the set of in-neighbors of node $i$. Moreover, if $j\in {\mathcal{N}}_{i}$, then $a_{ij}=1$ and $a_{ii}=0$ otherwise. The associated Laplacian matrix $\mathcal{L}={{[{{l}_{ij}}]}_{N\times N}}$ is defined as ${{l}_{ii}}=\sum\limits_{k\in {{\mathcal{N}}_{i}}}{{{a}_{ik}}}$ and ${{l}_{ij}}=-{{a}_{ij}}$ for $(j,i)\in \mathcal{E}$.


The following assumption and lemma pertaining to the communication graph will be utilized later:

\begin{assumption}\label{ass1}
    The communication graph $\mathcal{G}(\mathcal{V} ,\mathcal{E},\mathcal{A} )$ contains a directed spanning tree.
\end{assumption}
\begin{lemma}[\cite{ren2005consensus}]\label{lem1}
    Zero is a simple eigenvalue of the Laplacian matrix $\mathcal{L}$ if and only if the graph $\mathcal{G}$ contains a spanning tree, and $\mathcal{L}\mathbf{1}=0$.
\end{lemma}

To address the consensus problem of the MAS described in \eqref{e1}, one can use the distributed finite-horizon LQR with 
a typical cost function in the following form 
\begin{align}\label{e2}
    J_i(x(t))&\triangleq \sum\limits_{k=0}^{T-1}\Biggl(u_i^{\top}(t+k)R_iu_i(t+k)+\!\!\sum\limits_{j\in \mathcal{N}_i}a_{ij} (x_i(t+k)\Biggr.\nonumber\\
    &\Biggl.\quad-x_j(t))^{\top}Q_i(x_i(t+k)-x_j(t))\Biggr),
\end{align}
where $Q_i \in \mathcal{Q}_i \subset {\mathbb{S}^{n\times n}}$, $R_i\in \mathcal{R}_i \subset {\mathbb{S}_+^{n\times n}}$ are weight matrices designed to balance the consensus errors and energy consumption. Herein, $\mathcal{Q}_i$ and $\mathcal{R}_i$ represent the finite sets of $Q_i$ and $R_i$, respectively.


 
The overall objective is to achieve consensus with differential privacy guarantees for the quadratic cost function. This approach prevents agents from minimizing their own costs at the expense of others or causing system divergence. We specifically focus on protecting the pairwise weighting matrices $(Q_i, R_i)$, which effectively capture the operational priorities encoded in the cost function.




Let $W=\{(Q_i, R_i)\}_{i=1}^{N}$ denote the database that stores, for every agent $i$, the pair consisting of its weight matrices $Q_i$ $R_i$ of the quadratic cost. Define the set of all admissible databases as $\mathcal{W}=\{W\ |\ \forall Q_i\in \mathcal{Q}_i,\ \forall R_i\in \mathcal{R}_i,\ i=1,2,\ldots,N\}$. Additionally, we define a randomized mechanism $\mathbf{\hat{\mathcal{M}}}$ that maps a database $W$ to a random variable $\mathbf{\hat{\mathcal{M}}}(W)$. $\mathbf{\hat{\mathcal{M}}}(W)$ is typically the transmitted data (agents' states). Let $\text{range}(\mathbf{\hat{\mathcal{M}}})$ denote the set of all possible outputs of $\mathbf{\hat{\mathcal{M}}}(W)$. The adjacency relation between databases and the differential privacy of a mechanism $\mathbf{\hat{\mathcal{M}}}$ is subsequently defined.

\begin{definition}\label{de3}
    Two databases $W=\{(Q_i, R_i)\}_{i=1}^{N}$ and $W'=\{(Q'_i, R'_i)\}_{i=1}^{N}$ are adjacent if there exists $i\in \{1,2,\ldots,N\}$ {such that for all $j\neq i$, the pairs satisfy 
    \begin{multline*}
      \left(\frac{Q_j}{\mathrm{tr}(Q_j+R_j)}\!=\!\frac{Q'_j}{\mathrm{tr}(Q'_j+R'_j)}\right)\\ \land \left(\frac{R_j}{\mathrm{tr}(Q_j+R_j)}\!=\!\frac{R'_j}{\mathrm{tr}(Q'_j+R'_j)}\right)  
    \end{multline*}}
\end{definition}
Note that we aim to protect each weighting-matrix pair $(Q_i,R_i)$ as a single unit rather than hiding $Q_i$ and $R_i$ separately. Since the LQR feedback gain is scale-invariant (i.e., it depends only on their relative weighting), Definition \ref{de3} rescales every pair to a canonical form before declaring two databases adjacent.

\begin{definition}[\cite{dwork2006calibrating}]\label{de4}
    A mechanism $\mathbf{\hat{\mathcal{M}}}$ is $\epsilon$-differentially private if for any adjacent databases $W,W'\in \mathcal{W}$ and any set $\mathcal{S}\subset{\text{range}(\mathbf{\hat{\mathcal{M}}})}$, it holds that
    \begin{equation}\label{e3}
        \mathbb{P}\left(\mathbf{\hat{\mathcal{M}}}(W) \in \mathcal{S}\right)\leq \mathrm{e}^\epsilon\mathbb{P}\left(\mathbf{\hat{\mathcal{M}}}(W') \in \mathcal{S}\right)
    \end{equation}
\end{definition}
where $\epsilon>0$ is called the privacy budget, and a smaller $\epsilon$ implies a higher privacy level of $\mathbf{\hat{\mathcal{M}}}$. Combining Definitions \ref{de3} and \ref{de4}, differential privacy characterizes the sensitivity of the mechanism $\mathbf{\hat{\mathcal{M}}}$ with respect to the pairwise weight matrices $(Q_i, R_i)$. 

Let $\mathbf{\mathcal{M}}_t$ denote the finite-horizon LQR control for agents $i=1,2,\ldots,N$, at time $t$. The $\mathbf{\mathcal{M}}_t(W)$ represents the transmitted information $x_i(t+1)$ and $\text{range}(\mathbf{\mathcal{M}}_t)$ then denotes all the possible sets containing agents' states $x_i(t+1),\ i=1,2,\ldots,N$. We adopt the Laplacian mechanism of differential privacy for $\mathbf{\mathcal{M}}_t$ as follows.
\begin{lemma}[\cite{dwork2014algorithmic,cortes2016differential}]\label{lem2}
    For a given mechanism $\mathbf{\mathcal{M}}_t$ with $\text{range}(\mathbf{\mathcal{M}}_t)=\mathbb{R}^n$, let $$\triangle=\max\limits_{W,W'}\|\mathbf{\mathcal{M}}_t(W)-\mathbf{\mathcal{M}}_t(W')\|_\infty$$ be the sensitivity of $\mathbf{\mathcal{M}}_t$, where $W$ and $W'$ are adjacent databases. Then the mechanism $\hat{\mathbf{\mathcal{M}}}_t(W)=\mathbf{\mathcal{M}}_t(W)+\eta$ preserves $\epsilon$-differential privacy in Definition \ref{de4}, where $\eta \in \mathbb{R}^n$ and each component of $\eta$ is an i.i.d. Laplace noise subject to $\text{Lap}(n\triangle/\epsilon)$.
\end{lemma}

The problem can now be defined as:
\begin{problem}\label{pro1}
Design a 
finite-horizon LQR control mechanism such that the MAS in \eqref{e1} can achieve bounded consensus, satisfying
$\mathbb{E}(\|x_i(t)\|^2) < \infty$ for all $t\ge0$. Additionally, there exists a scalar $\sigma>0$ such that $\lim\limits_{t\rightarrow\infty}\sum\limits_{i=1}^N\mathbb{E}(\|x_i(t)-\frac{1}{N}\sum\limits_{j=1}^Nx_j(t)\|^2) \leq \sigma$, $\forall i\in \mathcal{V}$. 
Meanwhile, the controller can preserve $\epsilon$-differential privacy under the adjacency relation described in Definition \ref{de3}.
\end{problem}

\section{Main Results}

\subsection{Differentially Private LQR Consensus Mechanism}
In view of Lemma \ref{lem2}, 
the cost function in \eqref{e2} is designed as follows
\begin{align}\label{e4}
    J_i(&x(t),c(t))\triangleq \sum\limits_{k=0}^{T-1}\Biggl(u_i^{\top}(t\!+\!k)R_iu_i(t\!+\!k)+c(t)\sum\limits_{j\in \mathcal{N}_i}a_{ij}\Biggr.\nonumber\\
    &\Biggl.\quad  (x_i(t+k)-\hat{x}_j(t))^{\top}Q_i(x_i(t+k)-\hat{x}_j(t))\Biggr),
\end{align}
where $c(t) \in \mathbb{R}_{+}$ is a sequence to be designed, $\hat{x}_j(t)$ is 
the perturbed state received from the neighbors, defined as
\begin{equation}\label{eq:noisyinfor}
  \hat{x}_j(t)\triangleq x_j(t)+\eta_j(t),
\end{equation}
with $\eta_j(t)$ the injected Laplace noise for privacy guarantees, as will be elaborated later in Section~\ref{subsec:dp}.
An optimal control sequence $U^\ast_i$ is obtained by minimizing $J_i(x(t),c(t))$, with the control input $u_i(t)=u_i^{\ast}(t)$. We now proceed to present the proposed Algorithm 1 to solve Problem \ref{pro1}. 

Within this framework, we treat Algorithm 1 as a sequence of LQR mechanisms: $\mathbf{\mathcal{M}}=$ $\{\mathbf{\mathcal{M}}_0,\mathbf{\mathcal{M}}_1,\mathbf{\mathcal{M}}_2,\ldots\}$, where $\mathbf{\mathcal{M}}_t$ denotes the execution of Steps 4-9 at time $t$. {The selection of $c(t)$ and $p(t)$ will be clarified in the following subsections.}

\begin{algorithm}
\caption{Differentially Private distributed finite-horizon LQR-based Consensus Algorithm}
\begin{algorithmic}[1]
\STATE \textbf{Initialization:} $x_i(0)\in \mathbb{R}^n$, $\eta_i(0)=\mathbf{0}_n$, $\epsilon \in\mathbb{R}_+$, feasible sets $\mathcal{Q}_i$ and $\mathcal{R}_i$, $i=1,2,\ldots,N$.
\STATE Each agent $i$ {selects} $Q_i$ and $R_i$ from the feasible sets.
\STATE \textbf{Set} {time series} $c(t)$, $p(t)$ for all agents.
\FOR{$t=0, 1, 2, \ldots$}
    \FOR{$i=1, 2, \ldots, N$}
    \STATE Agent $i$ sends $\hat{x}_i(t)=x_i(t)+\eta_i(t)$ to its neighbors, and computes $\triangle_{i,t}$ by \eqref{e7}.
    \ENDFOR \ $i=N$
    \STATE Each agent $i$ obtains $u_i(t)$ by minimizing the $J_i$ in \eqref{e4} and updates the state $x_i(t+1)$ by \eqref{e1}. 
    \STATE Each agent $i$ generates noise $\eta_i(t+1)$ for the next iteration, with each component independently drawn from \eqref{eq:laplace}.
    \ENDFOR
\end{algorithmic}
\end{algorithm}

\subsection{Privacy Analysis}\label{subsec:dp}
In this subsection, we prove that when $c(t)$ and $p(t)$ in Step 3 are properly designed, then Algorithm 1 is differentially private. We begin by theoretically analyzing the solution to the minimization of \eqref{e4}.

Combining \eqref{e1} and \eqref{e4}, it follows that the unconstrained control input $u_i(t)$ (the first item of the optimal control sequence $U_i^{\ast}$), by explicit MPC \cite{bemporad2002model}, can be represented in the following
\begin{align}\label{eqinput}
    u_i(t) = c(t)K_{Q_i,R_i,t}\sum\limits_{j\in\mathcal{N}_i}a_{ij}  (\hat{x}_j(t)-x_i(t)),
\end{align}where $K_{Q_i,R_i,t}$ is a time-varying gain matrix depending on $Q_i$ and $R_i$.  One can refer to \cite{bemporad2002model,seron2000global} to compute the matrix $K_{Q_i,R_i,t}$. 
If $\mathcal{Q}_i, \mathcal{R}_i$ are bounded, $K_{Q_i,R_i,t}$ remains bounded for all $t\in\mathbb{N}$.
The closed-loop system dynamics for agent $i$ under \eqref{eqinput} are given by
\begin{equation}\label{e6}
    x_i(t+1)=x_i(t)+c(t)\sum\limits_{j\in \mathcal{N}_i}a_{ij}K_{Q_i,R_i,t}(\hat{x}_j(t)-x_i(t)).
\end{equation} We further define 
\begin{align}\label{e7}
    \triangle_{i,t}&\triangleq\max\limits_{W,W'}\{\|K_{Q_i,R_i,t}-K_{Q'_i,R'_i,t}\|_\infty:\ (Q_i,R_i)\in W,\nonumber\\ &\quad (Q'_i,R'_i)\in W', W\ \text{and}\ W'\ \text{satisfy Definition \ref{de3}}\}.
\end{align}
It is straight forward to show that $\triangle_{i,t}$ is bounded for all $t\in\mathbb{N}$. 
The Laplace noise $\eta_i(t)$ is then designed such that each component of $\eta_i$ is an independent and identically distributed Laplace noise with the consensus error-dependent scale
\begin{multline}\label{eq:laplace}
\hspace{-3mm}\eta_i(t+1) \!\sim\! \text{Lap}\!\left(\!nc(t)p(t){\triangle}_{i,t}\|\sum\limits_{j\in\mathcal{N}_i}\!\!a_{ij}[\hat{x}_j(t)\!-\!x_i(t)]\|_\infty/\epsilon \!\right)\!\!,  \\
t\in\mathbb{N},
\end{multline}
where $\eta_i(0)\!=\!\mathbf{0}_n$, $p(t) \in\mathbb{R}_+$ is a sequence to be co-designed with $c(t)$. 

\begin{assumption}\label{ass2}
$\sum\limits_{t=0}^{\infty}\frac{1}{p(t)}<\infty$.
\end{assumption}

Now, the differential privacy of the proposed algorithm is characterized.
\begin{theorem}\label{the1}
    Suppose Assumption \ref{ass1} holds and $p(t)$ is designed such that Assumption~\ref{ass2} is fulfilled.  Algorithm 1 ensures that $\mathbf{\mathcal{M}}_t$ preserves $\frac{\epsilon}{p(t)}$-differential privacy at time $t$. Algorithm 1 ($\mathbf{\mathcal{M}}$) preserves $\sum\limits_{t=0}^{\infty}\frac{\epsilon}{p(t)}$-differential privacy when the number of iterations tends to infinity.
\end{theorem}
\begin{proof}
    We first analyze the privacy property of $\mathbf{\mathcal{M}}_t$ at time $t$. Given the adjacent databases $W$ and $W'$ and the resulting  state vectors {$\mathbf{\mathcal{M}}_t(W)=\left[\begin{matrix}
        x_1^{\top}(t+1)  & x_2^{\top}(t+1) & \cdots& x^{\top}_N(t+1)
\end{matrix}\right]^{\top}$, $\mathbf{\mathcal{M}}_t(W')=\left[\begin{matrix}
        x'^{\top}_1(t+1) & x'^{\top}_2(t+1) & \cdots& x'^{\top}_N(t+1)
\end{matrix}\right]^{\top}$, }
where $x_j(t+1)=x'_j(t+1), j\neq i$, 
Then, the sensitivity of $\mathbf{\mathcal{M}}_t$ is derived as 
\begin{align*}
    &\max\limits_{W,W'}\|\mathcal{M}_t(W)\!-\!\mathcal{M}_t(W')\|_\infty
    \\
    &\quad\leq c(t)\|\sum\limits_{j\in\mathcal{N}_i}a_{ij}[\hat{x}_j(t)-x_i(t)]\|_\infty\triangle_{i,t}
\end{align*}
where the ``$\leq$'' is obtained from equations \eqref{e6} and \eqref{e7}. Consequently, by Lemma \ref{lem2}, the differential privacy level $\epsilon_t$ of $\mathbf{\mathcal{M}}_t$ satisfies 
$$\epsilon_t=\frac{n\times\max\limits_{W,W'}\|\mathcal{M}_t(W)-\mathcal{M}_t(W')\|_\infty}{nc(t)p(t){\triangle}_{i,t}\|\sum\limits_{j\in\mathcal{N}_i}a_{ij}[\hat{x}_j(t)-x_i(t)]\|_\infty/\epsilon}\leq\frac{\epsilon}{p(t)}.$$ Thus, we claim that $\mathbf{\mathcal{M}}_t$ preserves at least $\frac{\epsilon}{p(t)}$-differential privacy at time $t$. Since $\text{range}(\mathbf{\mathcal{M}}_t(W))$ depends on $\text{range}(\mathbf{\mathcal{M}}_{0}(W)),\ \text{range}(\mathbf{\mathcal{M}}_{1}(W)),\ \ldots,\ \text{range}(\mathbf{\mathcal{M}}_{t-1}(W))$, the Adaptive Sequential Composition Theorem \cite{dwork2014algorithmic} implies that Algorithm 1 ($\mathbf{\mathcal{M}}$) preserves $\sum_{t=0}^\infty \frac{\epsilon}{p(t)}$-differential privacy as $t$ goes to infinity. The proof is completed.
\end{proof}
\begin{remark}
   Note that $K_{Q_i,R_i,t}$ in \eqref{eqinput} is a solution without constraints. For the case with polyhedral constraints by \cite{bemporad2002model}, the proposed Algorithm 1 remains effective as long as the computed $K_{Q_i,R_i,t}$ are bounded for all $t\in\mathbb{N}$. The detailed discussion of cases with constraints is beyond the scope of this paper but will be presented in future research.
\end{remark}

\subsection{Consensus Analysis}
Now, we analyze the consensus performance of Algorithm 1 and examine its trade-off with the privacy level. Let $P\triangleq \mathbf{I}_N-\frac{1}{N}\mathbf{1}\mathbf{1}^{\top}$ and define $\delta(t)\triangleq (P\otimes \mathbf{I}_n)X(t)$ as the consensus error, where $X(t)=\left[\begin{matrix}
    x_1^{\top}(t) & x_2^{\top}(t)&\ldots&x_N^{\top}(t)
\end{matrix}\right]^{\top}$ and $\delta(t)=\left[\begin{matrix}
    \delta_1^{\top}(t) & \delta_2^{\top}(t)&\ldots&\delta_N^{\top}(t)
\end{matrix}\right]^{\top}$ with $\delta_i\in\mathbb{R}^n,\ i=1,2,\ldots,N$.

\begin{lemma}\label{lem4}
    For a sequence $a_n$ satisfying 
    \begin{equation}\label{eq:a_n}
     a_n\leq (1+c_n)a_{n-1}+d_n\sum\limits_{k=1}^n a_{n-k},   
    \end{equation}
    where $c_n\ge 0, d_n\ge 0$, $\forall\, n\in \mathbb{N}$, $a_n$ is bounded for all $ n\in \mathbb{N}$ if $\sum\limits_{n=1}^\infty c_n+nd_n<\infty$.
\end{lemma}
\begin{proof}
    Assume that there exists a monotonically increasing sequence $\beta_n$ such that $a_n\leq \beta_n,\,\forall n\ge0$. According to \eqref{eq:a_n}, we have $a_n\leq (1+c_n)\beta_{n-1}+d_n\sum\limits_{k=1}^n \beta_{n-k}\leq(1+c_n+nd_n)\beta_{n-1}$, which is fulfilled if $\beta_n$ follows $\beta_n=(1+c_n+nd_n)\beta_{n-1}$. In such a case,
    \begin{align*}\beta_n=\mathrm{e}^{\sum\limits_{t=1}^{n}\ln(1+c_t+td_t)}\beta_0.
    \end{align*} 
    Due to the fact that $\ln(1+c_t+td_t)\leq c_t+td_t$ provided $c_t+td_t>0$, it holds that 
        \begin{align*}
\beta_n&=\mathrm{e}^{\sum\limits_{t=1}^{n}\ln(1+c_t+td_t)}\beta_0\leq \mathrm{e}^{\sum\limits_{t=1}^{n}c_t+td_t}\beta_0 \leq \mathrm{e}^{\sum\limits_{t=1}^{\infty}c_t+td_t}\beta_0.
    \end{align*}
    Given $\sum_{n=1}^\infty c_n+nd_n<\infty$, it is immediate to show that 
$\beta_n$ is bounded for all $n\in\mathbb{N}$, which, in turn, implies the boundedness of $a_n,\,\forall n\in\mathbb{N}$. The proof is completed.
\end{proof}

We now impose an additional assumption on $c(t)$ and $p(t)$.
\begin{assumption}\label{ass3}
$\sum\limits_{t=0}^\infty c^2(t)=C_1<\infty$, $\kappa_{\epsilon} c^2(t)p^2(t)\leq1, \forall t \in \mathbb{N} $ and $\sum\limits_{t=1}^\infty tc^2(t)c^2(t-1)p^2(t-1)=C_2< \infty$.
\end{assumption}
where $\kappa_{\epsilon}$ is a constant influenced by the privacy budget.
\begin{align}\label{ead1}
\kappa_{\epsilon}\triangleq\frac{4n^3N(N-1)\hat{\triangle}^2}{\epsilon^2}.
\end{align}
with \begin{equation}\label{e9}
    \hat{\triangle}\triangleq\max_{i=1,\ldots,N}\hat{\triangle}_i,\ \hat{\triangle}_i\triangleq\sup\limits_{t=0,1,\ldots}\triangle_{i,t}.
\end{equation} 
\begin{theorem}\label{the2}
    Given system \eqref{e1} satisfying Assumption \ref{ass1}, and $c(t)$ and $p(t)$ satisfying Assumptions~\ref{ass2}-\ref{ass3}. the second moment of the consensus error, $\mathbb{E}[\delta^{\top}(t)\delta(t)]$, is bounded $\forall t\in \mathbb{N}$ under Algorithm 1.
\end{theorem}
\begin{proof}
Consider $\bar{\mathcal{L}}=\mathcal{L}\otimes\mathbf{I}_n,\bar{\mathcal{A}}=\mathcal{A}\otimes\mathbf{I}_n,\bar{P}=P\otimes\mathbf{I}_n$. From \eqref{eq:noisyinfor} and \eqref{e6}, it can be shown that 
\begin{align}\label{e10}
    X(t+1)=\left(\mathbf{I}_{nN}-c(t)K_t\bar{\mathcal{L}}\right)X(t)+c(t)K_t\bar{\mathcal{A}}\eta(t),
\end{align}where $\eta(t)\!=\!\left[\begin{matrix}
        \eta_1(t)^{\top} & \eta_2(t)^{\top} & \ldots&\eta_N(t)^{\top}
    \end{matrix}\right]^{\top}$,
    \begin{align*}
    K_t\!\!=\!\!\left[\begin{matrix}
        K_{Q_1,R_1,t} & \mathbf{0}_{n\times n} &\ldots &\mathbf{0}_{n\times n}\\\mathbf{0}_{n\times n}&K_{Q_2,R_2,t}&\ldots&\mathbf{0}_{n\times n}\\
        \vdots&\vdots&\ddots&\vdots\\\mathbf{0}_{n\times n}&\mathbf{0}_{n\times n}&\ldots& K_{Q_N,R_N,t}
    \end{matrix}\right].
\end{align*}
Define $\bar{\triangle}_i\triangleq\sup\limits_{{t=0,1,\ldots}}\|K_{Q_i,R_i,t}\|$ and we can find a scalar
\begin{equation}\label{e13}
\bar{\triangle}=\max\limits_{i=1,\ldots,N}\bar{\triangle}_i\ge \max\limits_{i=1,\ldots,N}\|K_{Q_i,R_i,t}\|=\|K_t\|,
\end{equation}for all $t\in\mathbb{N}$. From Lemma \ref{lem1} and the fact that  
$\bar{\mathcal{L}}\delta(t)=(\mathcal{L}P\otimes\mathbf{I}_n)X(t)=\bar{\mathcal{L}}X(t)$, the dynamics of the consensus error follow
\begin{align}\label{e11}
     \delta(t+1)=\left(\mathbf{I}_{nN}-c(t)\bar{P}K_t\bar{\mathcal{L}}\right)\delta(t)+c(t)\bar{P}K_t\bar{\mathcal{A}}\eta(t).
\end{align}
Consider a Lyapunov candidate $V(t)\triangleq\delta^{\top}(t)\delta(t)$. By using the identity $PP^\top=P^2=P$, it can be shown that
\begin{align}\label{e12}
    V(t+1)&= V(t)-2c(t)\delta^{\top}(t)\frac{\bar{P}K_t\bar{\mathcal{L}}+(\bar{P}K_t\bar{\mathcal{L}})^{\top}}{2}\delta(t)\nonumber\\
    &\ \ \, \ +c^2(t)\delta^{\top}(t)\bar{\mathcal{L}}^{\top}K_t^{\top}\bar{P}K_t\bar{\mathcal{L}}\delta(t)\nonumber\\
    &\ \ \, \ +c^2(t)\eta^{\top}(t)\bar{\mathcal{A}}^{\top}K_t^{\top}\bar{P}K_t\bar{\mathcal{A}}\eta(t)\nonumber\\
     &\ \ \, \ +2c(t)\eta^{\top}(t)\bar{\mathcal{A}}^{\top}K_t^{\top}\bar{P}(\mathbf{I}_{nN}-c(t)\bar{P}K_t\bar{\mathcal{L}})\delta(t).
\end{align} From \eqref{eq:laplace}, $\mathbb{E}[\eta(t)]=0$ holds under all circumstances for all $t\in\mathbb{N}$. Combined with the Double Expectation Theorem \cite{billingsley2017probability}, we have $\mathbb{E}[\eta^{\top}(t)\bar{\mathcal{A}}^{\top}K_t^{\top}\bar{P}(\mathbf{I}_{nN}-c(t)\bar{P}K_t\bar{\mathcal{L}})\delta(t)]=\mathbb{E}[\mathbb{E}[\eta^{\top}(t)\bar{\mathcal{A}}^{\top}K_t^{\top}\bar{P}(\mathbf{I}_{nN}-c(t)\bar{P}K_t\bar{\mathcal{L}})\delta(t)|\delta(t)]]=0$. For the third term on the right-hand side of \eqref{e12}, it can be easily shown that
\begin{align*}
    \mathbb{E}[c^2(t)\delta^{\top}(t)\bar{\mathcal{L}}^{\top}K_t^{\top}\bar{P}K_t\bar{\mathcal{L}}\delta(t)]\leq c^2(t)\rho_{\mathcal{L}}\bar{\triangle}^2\mathbb{E}[V(t)],
\end{align*} 
where {$\rho_{\mathcal{L}}=\|\bar{\mathcal{L}}\|^2$}.

Now, we proceed to discuss the fourth term. Similar to the third term, we have
\begin{align*}
    \mathbb{E}[c^2(t)\eta^{\top}(t)\bar{\mathcal{A}}^{\top}K_t^{\top}\bar{P}K_t\bar{\mathcal{A}}\eta(t)]\leq c^2(t)\rho_{\mathcal{A}}\bar{\triangle}^2\mathbb{E}[\eta^{\top}(t)\eta(t)],
\end{align*}where $\rho_{\mathcal{A}}=\|\bar{\mathcal{A}}\|^2$. For the sake of further discussion, let $\eta_{ij}\in\mathbb{R},\, j=1,2,\ldots,n,$ be the $j$th element of $\eta_i,\, i=1,2,\ldots,N$, which, by \eqref{eq:laplace}, are mutually independent Laplace noise. Then
\begin{multline}\label{e14}
   \mathbb{E}[\eta_i^{\top}(t)\eta_i(t)]=\sum\limits_{j=1,\ldots,n}\mathbb{E}\left[\eta_{ij}^{\top}(t)\eta_{ij}(t)\right]\\
   = \frac{2n^3c^2(t-1)p^2(t-1)}{\epsilon^2}\\
\times\hat{\triangle}^2_{i,t}\mathbb{E}\left[\|\sum\limits_{j\in\mathcal{N}_i}a_{ij}[\hat{x}_j(t-1)-x_i(t-1)]\|^2_\infty\right].
\end{multline}
Applying the triangle and the Cauchy–Schwarz inequalities yields
    \begin{align}\label{e15}
    &\mathbb{E}\left[\|\sum\limits_{j\in\mathcal{N}_i}a_{ij}[\hat{x}_j(t-1)-x_i(t-1)]\|^2_\infty\right]\nonumber\\
    & \leq 2(N-1)\Biggl[\sum\limits_{j\in\mathcal{N}_i}a_{ij}\mathbb{E} \left[\|{x}_j(t-1)-x_i(t-1)\|_\infty^2\right] \nonumber\Biggr.\\\Biggl.
    &\quad +\sum\limits_{j\in\mathcal{N}_i}a_{ij}\mathbb{E}\left[\|\eta_j(t-1)\|_\infty^2\right]\Biggr].
\end{align}
Consider the first additive item on the right-hand side of \eqref{e15}:
    \begin{align}\label{e16}
    &\sum\limits_{j\in\mathcal{N}_i}a_{ij}\mathbb{E} \left[\|{x}_j(t-1)-x_i(t-1)\|_\infty^2\right] \nonumber\\
    &\leq 2\sum\limits_{j\in\mathcal{N}_i}a_{ij}\mathbb{E} \left[\|{x}_j(t-1)-\frac{1}{N}\sum\limits_{k=1,\ldots,N}x_k(t-1)\|_\infty^2\right]\nonumber\\
    &\ \ \, \ +2\sum\limits_{j\in\mathcal{N}_i}a_{ij}\mathbb{E} \left[\|{x}_i(t-1)-\frac{1}{N}\sum\limits_{k=1,\ldots,N}x_k(t-1)\|_\infty^2\right]\nonumber\\
    & \leq 2(N-1)\mathbb{E}[V(t-1)].
\end{align} 
The second ``$\leq$'' holds since $\sum\limits_{j\in\mathcal{N}_i}a_{ij}\leq N-1, \forall i,$ and $V(t-1)=\sum\limits_{i=1}^{N}\|{x}_i(t-1)-\frac{1}{N}\sum\limits_{k=1,\ldots,N}x_k(t-1)\|^2$. The second additive item of inequality \eqref{e15} can be upper bounded by
    \begin{align}\label{e17}
        &\sum\limits_{j\in\mathcal{N}_i}a_{ij}\mathbb{E}[\|\eta_j(t-1)\|_\infty^2] \leq \mathbb{E}[\eta^\top(t-1)\eta(t-1)].
    \end{align}
Combining (\ref{ead1}-\ref{e9}) and (\ref{e14}-\ref{e17}), there holds
    \begin{align} \label{e18}
    \mathbb{E}[\eta^\top(t)\eta(t)]&\leq \kappa_{\epsilon} c^2(t-1)p^2(t-1)\Big[2(N-1)\mathbb{E}[V(t-1)]\Big.\nonumber\\
   \Big. & \quad+\mathbb{E}[\eta^\top(t-1)\eta(t-1)]\Big],
\end{align}
Note that $ \mathbb{E}[\eta^\top(0)\eta(0)]=0$ and unfolding the recursive relation in \eqref{e18}, we obtain
\begin{small}
    \begin{align} \label{noiseeq}
    \mathbb{E}[\eta^\top(t)\eta(t)]&\!\leq\!\sum\limits_{k=1}^t\!\left[2(N\!-\!1)\mathbb{E}[V(t\!-\!k)]\prod_{j=1}^k\kappa_{\epsilon} c^2(t\!-\!j)p^2(t\!-\!j)\right]\nonumber\\
&\leq \sum\limits_{k=1}^t 2(N-1)\kappa_{\epsilon} c^2(t-1)p^2(t-1)\mathbb{E}[V(t-k)].
\end{align}
\end{small}
The second ``$\leq$'' follows from Assumption~\ref{ass3}\footnote{
{The boundedness of noise variance is guaranteed if the two sequences further satisfy $\{c(t) p(t)\} \in \ell_2$.}}. 
Taking the expectation of both sides of \eqref{e12} and removing the term $-2c(t)\delta^{\top}(t)\frac{\bar{P}K_t\bar{\mathcal{L}}+(\bar{P}K_t\bar{\mathcal{L}})^{\top}}{2}\delta(t)$ due to its negativity, yields
    \begin{align}\label{e20}
    \mathbb{E}[V(t+1)]\!\leq & (1+ c^2(t)\rho_{\mathcal{L}}\bar{\triangle}^2)\mathbb{E}[V(t)]   +c^2(t)\rho_{\mathcal{A}}\bar{\triangle}^2\sum\limits_{k=1}^t\kappa_{\epsilon} \nonumber \\ &\times  2(N-1) c^2(t-1)p^2(t-1)\mathbb{E}[V(t-k)].
\end{align} 
Following Lemma \ref{lem4} and Assumption~\ref{ass3}, it can be concluded that $\mathbb{E}[V(t)]$ is bounded for all $t$. The proof is completed.
\end{proof}

It is known from Theorem \ref{the2} that there exists an upper bound for $\mathbb{E}[V(t)], \forall t\in\mathbb{N}$. Without loss of generality, we assume the upper bound is $\bar{V}$. To proceed to the final result, further assumption is required for $c(t)$ and $p(t)$. {
\begin{assumption}\label{ass4}
$c^2(t)\rho_{\mathcal{L}}\bar{\triangle}^2-c(t)\rho_{\lambda}>-1,\forall t\ \in \mathbb{N},$
\end{assumption}
\noindent where 
\begin{align}\label{e21}
    {\  \rho_{\lambda}\triangleq\inf\limits_{t=0,1,\ldots}\lambda_{0}(\bar{P}K_t\bar{\mathcal{L}}+(\bar{P}K_t\bar{\mathcal{L}})^{\top}).}
\end{align}}

\begin{theorem}\label{the3}
    Given system \eqref{e1} satisfying Assumption \ref{ass1} and $c(t)$ and $p(t)$ complying with Assumptions~\ref{ass2}-\ref{ass4}, Algorithm 1 ensures bounded consensus while preserving $\sum\limits_{t=0}^{\infty}\frac{\epsilon}{p(t)}$-differential privacy with the error bound
    \begin{align}\label{e22}
\sigma&=\frac{\mathrm{e}^{\rho_{\mathcal{L}}\bar{\triangle}^2C_1-\sum\limits_{t=0}^\infty c(t)\rho_{\lambda}}}{(1+c^2(0)\rho_{\mathcal{L}}\bar{\triangle}^2-c(0)\rho_{\lambda})}V(1)+2(N-1)\rho_{\mathcal{A}}\kappa_{\epsilon}\nonumber\\
       &\quad\times\bar{\triangle}^2 C_2 C_3\bar{V},
    \end{align} 
where $C_3$ is a positive constant that will be defined later. 
$V(1)$ is a constant determined by $V(0),c(0),\eta(0)$ (see \eqref{e12}). 
\end{theorem}
\begin{proof}
Given $\mathbb{E}[V(t)]$ is bounded, it follows from $\delta(t)$ $=\bar{P}X(t)$ that $\mathbb{E}[X^\top(t)X(t)]$ is bounded, which then implies the boundedness of $\mathbb{E}[\|x_i(t)\|^2]$.

Recall \eqref{e12}, \eqref{e21} {and refer to \cite{li2018distributed}}, there holds
\begin{align*}
    2c(t)\delta^{\top}(t)\frac{\bar{P}K_t\bar{\mathcal{L}}+(\bar{P}K_t\bar{\mathcal{L}})^{\top}}{2}\delta(t)\ge c(t)\rho_{\lambda}V(t).
\end{align*}
With reference to \eqref{e20}, it can then be obtained that
\begin{align*}
    \mathbb{E}[V(t+1)]&\leq \!(1+ c^2(t)\rho_{\mathcal{L}}\bar{\triangle}^2\!-\!c(t)\rho_{\lambda})\mathbb{E}[V(t)]   
   \!+\!2t(N\!-\!1)\\
   &\quad \times \rho_{\mathcal{A}} \kappa_{\epsilon}\bar{\triangle}^2 c^2(t) c^2(t-1)p^2(t-1)\bar{V}\\
   &=V(1)\!\!\prod_{k=1}^t(1\!+\! c^2(k)\rho_{\mathcal{L}}\bar{\triangle}^2\!-\!c(k)\rho_{\lambda})\!+\!2(N\!-\!1)\\
   &\quad\times\rho_{\mathcal{A}}\kappa_{\epsilon}\bar{\triangle}^2\bar{V}\sum\limits_{k=1}^{t}\Bigg[ kc^2(k) c^2(k-1)p^2(k-1)\Bigg.\\
   &\Bigg.\quad \times\prod\limits_{j=k}^t (1+ c^2(j)\rho_{\mathcal{L}}\bar{\triangle}^2-c(j)\rho_{\lambda})\Bigg].
\end{align*} 
{Under Assumption~\ref{ass4}, we have}
\begin{align*}
\lim\limits_{t\rightarrow\infty}\prod_{k=0}^t(1+ c^2(k)\rho_{\mathcal{L}}\bar{\triangle}^2-c(k)\rho_{\lambda}) \leq \mathrm{e}^{\rho_{\mathcal{L}}\bar{\triangle}^2C_1-\sum\limits_{k=0}^\infty c(k)\rho_{\lambda}}.
\end{align*} 
which implies the existence of $C_3$, such that
\begin{equation}\label{ec3}
    \prod\limits_{j=t}^\infty (1+ c^2(j)\rho_{\mathcal{L}}\bar{\triangle}^2\!-\!c(j)\rho_{\lambda})<C_3, \forall t\in\mathbb{N}.
\end{equation}
Therefore, the following inequality holds
\begin{multline*}
    \lim\limits_{t\rightarrow\infty} \sum\limits_{k=1}^{t} \left[ kc^2(k) c^2(k\!-\!1)p^2(k\!-\!1)
   \right.\\ \left.\prod\limits_{j=k}^t (1+ c^2(j) \rho_{\mathcal{L}}\bar{\triangle}^2\!-\!c(j)\rho_{\lambda})\right]\! 
   < {C_2C_3}.
   \end{multline*}
   We obtain $\mathbb{E}[V(t+1)]\leq\sigma$ as $t\rightarrow \infty$, which is given in \eqref{e22}. The proof is completed.
\end{proof}

Equation \eqref{e22} reveals the impact of the privacy level $\epsilon$ on the error bound $\sigma$. If $\epsilon\rightarrow \infty$, then $\kappa_{\epsilon}\rightarrow 0$ from \eqref{ead1}, leading to $2(N-1)\rho_{\mathcal{A}}\kappa_{\epsilon}\bar{\triangle}^2 C_2  \mathrm{e}^{\rho_{\mathcal{L}}\bar{\triangle}^2C_1-\sum_{k=0}^\infty c(k)\rho_{\lambda}}\bar{V}\rightarrow0$. Note that $\kappa_{\epsilon} c^2(t)p(t)^2\leq1$ from Assumption~\ref{ass3}, if $c(t)$ is designed such that ${\rho_{\mathcal{L}}\bar{\triangle}^2C_1-\sum_{t=0}^\infty c(t)\rho_{\lambda}}\leq0$, there holds ${\rho_{\mathcal{L}}\bar{\triangle}^2C_1-\sum_{t=0}^\infty c(t)\rho_{\lambda}}=-\mathit{O}(\frac{1}{\sqrt{\kappa_{\epsilon}}})\rightarrow -\infty$ as $\kappa_{\epsilon}\rightarrow 0$, and finally $$\frac{\mathrm{e}^{\rho_{\mathcal{L}}\bar{\triangle}^2C_1-\sum_{k=0}^\infty c(k)\rho_{\lambda}}}{(1+c^2(0)\rho_{\mathcal{L}}\bar{\triangle}^2-c(0)\rho_{\lambda})}V(1)\rightarrow 0.$$ 
On the other hand, \eqref{e22} indicates that $\epsilon$ cannot be selected too small; otherwise, $\mathbb{E}[V(t)]$ may exceed $V(0)$, which is usually not acceptable. A properly chosen $\epsilon$ should make $2(N-1)\rho_{\mathcal{A}}\kappa_{\epsilon}\bar{\triangle}^2 C_2  \mathrm{e}^{\rho_{\mathcal{L}}\bar{\triangle}^2C_1-\sum_{k=0}^\infty c(k)\rho_{\lambda}}$ sufficiently small and strictly less than $1$, with well-selected $Q_i$, $R_i$, $c(t)$ and $p(t)$. In such a case, $\bar{V}$ can be computed a priori and verified
\begin{align*}
\bar{V}=&\left[1-2(N-1)\rho_{\mathcal{A}}\kappa_{\epsilon}\bar{\triangle}^2 C_2  \mathrm{e}^{\rho_{\mathcal{L}}\bar{\triangle}^2C_1-\sum\limits_{k=0}^\infty c(k)\rho_{\lambda}}\right]^{-1}\\&\times\frac{\mathrm{e}^{\rho_{\mathcal{L}}\bar{\triangle}^2C_1-\sum\limits_{k=0}^\infty c(k)\rho_{\lambda}}}{(1+c^2(0)\rho_{\mathcal{L}}\bar{\triangle}^2-c(0)\rho_{\lambda})}V(1).
\end{align*}

\begin{remark}[Parameter choice]
    Note that $\kappa_{\epsilon} c^2(t)p^2(t)$ $\leq1, \forall t \in \mathbb{N}, $ in Assumption~\ref{ass3} is a sufficient condition for the validity of Theorem \ref{the3}. To achieve a smaller $\sigma$, it is reasonable to initialize the early terms of the sequence $c(t)$ at relatively large values (while complying with $\kappa_{\epsilon} c^2(t)p^2(t)\leq1$). 
    The impact of the weight matrices $Q_i$ and $R_i$ on the error bound $\sigma$ is multi-coupled and depends on the system parameters $n$ and $N$, as they simultaneously affect $\rho_{\lambda}$, $\bar{\triangle}$ and $\kappa_{\epsilon}$ in \eqref{e22}. This aspect will be further investigated in future work. Additionally, if the privacy budget $\sum\limits_{t=0}^{\infty}\frac{\epsilon}{p(t)}$ is not necessarily bounded (i.e., it can be understood that only finite-step privacy is required), then $c(t)$ can be selected such that $\sum_{t=0}^\infty c(t)\rightarrow\infty$, $\sum_{t=0}^\infty c^2(t)< \infty$, {and hence  $C_3\rightarrow 0$ for a reduced error bound}. It can be inferred from \eqref{e22} that $\sigma \rightarrow 0,\,t\rightarrow \infty$ and the system can achieve accurate mean-square consensus. 
\end{remark}

\section{Numerical Simulation}
Consider a 4-agent network in the form of \eqref{e1} with the communication graph given by 
\begin{align*}
    \mathcal{A}=\left[\begin{matrix}
        0 & 1 & 1 & 0\\0 & 0 &0 &1\\0 & 0 &0 &1\\1 & 0 & 0 &0
    \end{matrix}\right],\quad \mathcal{L}=\left[\begin{matrix}
        2 & -1 & -1 & 0\\0 & 1 &0 &-1\\0 & 0 &1 &-1\\-1 & 0 & 0 &1
    \end{matrix}\right].
\end{align*}The initial states are $x_1(0)=[1 \,10 \,20]^{\top}$, $x_2(0)=[14 \,20 \,7]^{\top}$, $x_3(0)=[30 \,70 \,50]^{\top}$, $x_4(0)=[55 \,43 \,34]^{\top}$.
Assign a privacy budget $\epsilon=5$, and set the rolling horizon length and $Q_i$, $R_i$ respectively to $T=10$ and  
\begin{align*}
    \mathcal{Q}_i=\left[\begin{matrix}
       q &0 & 0\\0 &q &0\\0 & 0& q
    \end{matrix}\right],  \mathcal{R}_i=\left[\begin{matrix}
      r &0 & 0\\0 &r &0\\0 & 0& r
    \end{matrix}\right],\ i=1,2,3,4.
\end{align*} 
where $q\in\{6,7,8\}$, $r\in\{2,3\}$ defines the feasible sets of $Q_i$, $R_i$.
Select 
\begin{align*}
    Q_i=\left[\begin{matrix}
       8 &0 & 0\\0 &8 &0\\0 & 0& 8
    \end{matrix}\right],  {R}_i=\left[\begin{matrix}
      2 &0 & 0\\0 &2 &0\\0 & 0& 2
    \end{matrix}\right],\ i=1,2,3,4.
\end{align*} Design $c(t)=\frac{1}{15(t+1)^{1.3}}$ and $p(t)={(t+1)^{1.1}}, \forall t\in\mathbb{N}$. The agents' states generated by Algorithm 1 are depicted in Fig.~\ref{fig.1}, which verifies its bounded consensus properties with guaranteed differential privacy. In particular, if an adversary continuously eavesdrops from time $t_1=0$ to $t_2=120$, the total privacy leakage is $\sum_{t=0}^{120}\frac{\epsilon}{p(t)}\approx21.9828$.

\begin{figure}[htp!]
 \centering
  \includegraphics[width=\columnwidth]{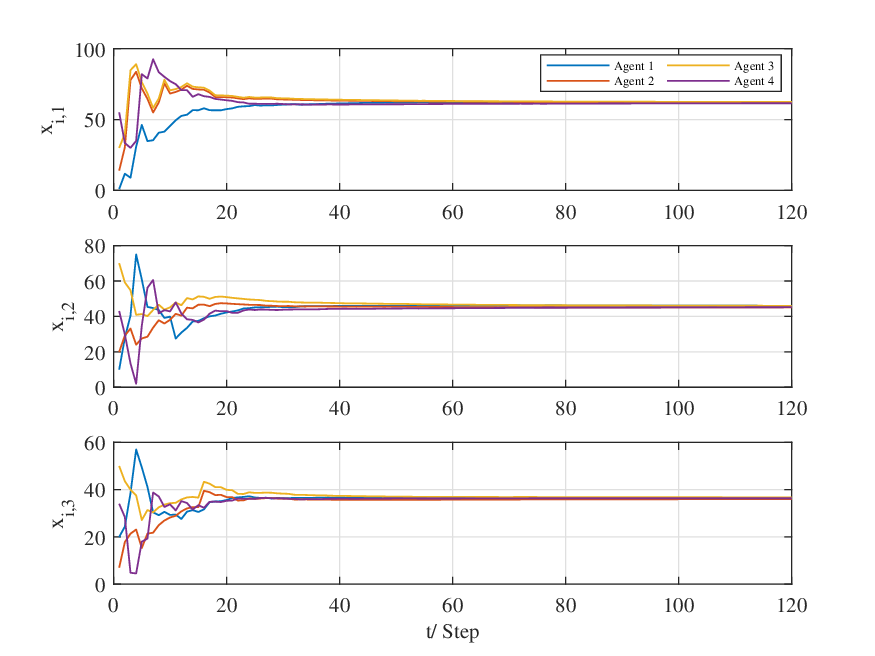}\\[-2.2ex]
  \caption{State trajectories of agents.}
  \label{fig.1}
\end{figure}

\section{Conclusion}
This paper proposes a distributed LQR consensus control algorithm with differential privacy guarantees. It is rigorously proven that under the proposed algorithm, both bounded consensus and differential privacy can be achieved without requiring additional assumptions on the system or the injected noises, other than graph connectivity. A numerical example is provided to show the effectiveness of the proposed protocol. Since the noise structure in this paper shares a similar form with the noise model inherent in the channel \cite{li2018distributed}, this raises the question of whether inherent channel noise itself could provide a certain level of differential privacy. This topic will be explored in our future research. 

\balance
\bibliographystyle{IEEEtran}
\normalem
\bibliography{list}

\end{document}